\documentclass[12pt,a4paper]{article}

\usepackage{amsthm}
\usepackage{amsmath}
\usepackage{amssymb}
\usepackage{amsfonts}
\usepackage{graphicx}

\newcommand\cotimes{\widehat{\otimes}}
\renewcommand\H{\mathcal{H}}
\newcommand\D{\mathcal{D}}
\renewcommand\L{\mathcal{L}}
\newcommand{\scalar}[2]{\langle #1\,,#2 \rangle}
\newcommand{\norm}[1]{|\negthinspace|#1|\negthinspace|}
\renewcommand{\d}{\operatorname{d}\negthinspace}
\newcommand{\N}{\mathcal{N}}
\renewcommand{\o}{\otimes}
\newcommand{\up}{\uparrow}
\newcommand{\down}{\downarrow}

\newtheorem{theorem}{Theorem}

\title{Boundary Dynamics Driven Entanglement}

\date{\today}

\author{
		A. Ibort\footnote{On leave of absence from: Depto.\ de Matem\'aticas, Univ. Carlos III de Madrid, Avda.\ de la
Universidad 30, 28911 Legan\'es, Madrid, Spain. This work has been partially supported by the Spanish MICIN grant
MTM 2010-21186-C02-02, QUITEMAD P2009 ESP-1594, and Fundaci\'on Caja Madrid.} \\ Dept.\ of Mathematics, Univ.\ of California at Berkeley\\ Berkeley CA 94720, USA \\ albertoi@math.uc3m.es
		\and 
		G. Marmo\footnote{G.M.\ would like to acknowledge the support provided by the Santander/UCIIIM Chair of Excellence programme 2011-2012.} \\ Dipartimento di Fisica, Universit\`a di Napoli ``Federico II'' \\ INFN-Sezione di Napoli\\ Via Cintia Edificio 6, I--80126 Napoli, Italy \\ marmo@na.infn.it
		\and
		J.M.\ P\'erez-Pardo\footnote{J.M.P.\ is supported by UCIIIM University through the PhD Program Grant M02-0910, 2012 grant for ``Movilidad de investigadores de la UCIIIM'', Spanish MICIN grant MTM 2010-21186-C02-02, QUITEMAD P2009 ESP-1594.  Thanks to the Dept.\ of Mathematics at UC Berkeley for its hospitality. Special thanks to Dipt.\ di Fisica at Universit\`a di Napoli ``Federico II'' for its support and hospitality.} \\ Depto.\ de Matem\'aticas, Univ.\ Carlos III de Madrid\\ Avda. de la Universidad 30, 28911 Legan\'es, Madrid, Spain \\ jmppardo@math.uc3m.es
		}

\begin{document}

%\author{A. Ibort\footnote{On leave of absence from: Depto.\ de Matem\'aticas, Univ. Carlos III de Madrid, Avda.\ de la
%Universidad 30, 28911 Legan\'es, Madrid, Spain. This work has been partially supported by the Spanish MICIN grant
%MTM 2010-21186-C02-02, QUITEMAD P2009 ESP-1594, and Fundaci\'on Caja Madrid.}}
%\affiliation{Dept.\ of Mathematics, Univ.\ of California at Berkeley, Berkeley CA 94720, USA}
%\email{albertoi@math.uc3m.es}
%\author{G. Marmo\footnote{G.M.\ would like to acknowledge the support provided by the Santander/UCIIIM Chair of Excellence programme 2011-2012.}}
%\affiliation{Dipartimento di Fisica, Universit\`a di Napoli ``Federico II'' and INFN-Sezione di Napoli, Via Cintia Edificio 6, I--80126 Napoli, Italy}
%\email{marmo@na.infn.it}
%\author{J.M. P\'erez-Pardo\footnote{J.M.P.\ is supported by UCIIIM University through the PhD Program Grant M02-0910, 2012 grant for ``Movilidad de investigadores de la UCIIIM'', Spanish MICIN grant MTM 2010-21186-C02-02, QUITEMAD P2009 ESP-1594.  Thanks to the Dept.\ of Mathematics at UC Berkeley for its hospitality. Special thanks to Dipt.\ di Fisica at Universit\`a di Napoli ``Federico II'' for its support and hospitality.}}
%\affiliation{Depto.\ de Matem\'aticas, Univ.\ Carlos III de Madrid, Avda. de la Universidad 30, 28911 Legan\'es, Madrid, Spain.}
%\email{jmppardo@math.uc3m.es}

%\pacs{03.65.Ca, 03.65.Db, 03.67.Bg, 02.30.Sa, 02.40.Ky, 02.40Vh, 2.30Yy}
%\keywords{Boundary Control, Self-Adjoint Extensions, Non-separable Boundary Conditions, Boundary Generated Entanglement}

\maketitle

\newpage
\begin{abstract}  
We will show how it is possible to generate entangled states out of unentangled ones on a bipartite system by means of dynamical boundary conditions.
The auxiliary system is defined by a symmetric but not self-adjoint Hamiltonian and the space of self-adjoint extensions of the bipartite system is studied.
It is shown that only a small set of them leads to separable dynamics and they are characterized.
Various simple examples illustrating this phenomenon are discussed,  in particular we will analyze the hybrid system consisting of a planar quantum rotor and a spin system under a wide class of boundary conditions.
\end{abstract}


\tableofcontents

%%%%%%%%%%%%%%%%%%%%%%%%%%%%%%%%%%%%%%%%%%%%%%%%%%%%%%%%%%%%%%%%%%%%%%%%%%%%%%%%%%%%%%%%%%%%%%%%%%%%%%%%%%%%%%%%%%%%%%%%%%%%%%%%%%%%%%%%%%%%%%%%%%%%%%%%%%%%%%%%%%%%%%%%%%%%%%%%%%%%%%%%%%%%%%%%%%%%%%%%%%%%%%%%%%%%%%%%%%%%%%%%%%%%%%%%%%%%%%%%%%%%%%%%
\section{Introduction}  There is an increasing interest in the physics associated to the ``boundary'' of a given physical system.   Because the boundary can be thought as an effective way of describing the interaction of the system with the external universe, its modeling could account for a number of significant physical phenomena.  

It is impossible to cover the range of physics associated to boundary structures in a few sentences.  We mention here
Casimir's effect, which is arguably one of the most conspicuous physical phenomena associated to the presence of boundaries (see for instance \cite{As06}, \cite{As08} and references therein for an extensive account of the role of boundary conditions and vacuum structures), the quantum Hall effect \cite{Mo88} and the appearance of Edge States \cite{As13}.   We would also like to mention here the possibility of describing topology change as a boundary effect.  This idea was already considered in \cite{Ba95} and further elaborated in relation with specific boundary conditions in \cite{As05}, but it has gained new impetus because of Wilczek's \emph{et al} \cite{Wi12} recent contributions to it.

In this paper we will explore how the manipulation of boundary conditions of composite systems allows to generate entangled states.    More precisely, consider two systems $A$ and $B$,  and assume that system $B$, which will be called the ``bulk'' or controlled system, is complete, i.e., its Hamiltonian $H_B$ is a Hermitean (self-adjoint) operator on a Hilbert space $\mathcal{H}_B$ and its evolution $U^B_t = \exp (itH_B)$ is unitary.  However the system $A$, or ``auxiliary'', is defined by a merely symmetric operator $H_A$ on a Hilbert space $\mathcal{H}_A$. In other words the evolution ``$U_t^A = \exp (itH_A)$'' will not be unitary until we have selected (if it exists) a self-adjoint extension of the operator $H_A$.   It is worth to point out here that such situation will actually arise whenever our system $A$ is defined in a bounded domain $\Omega_A$ in $\mathbb{R}^n$ with boundary $\partial \Omega_A$.   In such case the Hilbert space $\mathcal{H}_A$ is the space $\L^2(\Omega)$ of square integrable complex-valued functions on $\Omega_A$ and the Hamiltonian operator is
\begin{equation}\label{free}
H_A = -\frac{\hbar^2}{2m} \Delta_\eta  + V ,
\end{equation}
with $\Delta_\eta$ the Laplace-Beltrami operator defined by some metric $\eta$ on $\Omega_A$, and $V$ a potential function.   Under such circumstances it can be shown that the self-adjoint extensions of $H_A$ are determined by boundary conditions satisfied by the functions on the corresponding domain \cite{As05}.

The main observation which is relevant for the purposes of this paper is that, if we consider now the bipartite system defined on the Hilbert space $\mathcal{H}_A \hat\otimes \mathcal{H}_B$, the family of self-adjoint extensions of the symmetric Hamiltonian $H = H_A \otimes I + I \otimes H_B$ is much larger than the family of self-adjoint extensions of the standalone symmetric Hamiltonian $H_A$.   As it will be discussed along the paper, many of the possible self-adjoint extensions of the bipartite system generate entangled states out of separable ones.  In other words, the dynamics defined by many self-adjoint extensions of the composite system are not separable, i.e., they do not preserve separable states.

Separable dynamics for a class of hybrid composite systems will be characterized and it will be shown that they correspond to boundary conditions  defined by the tensor product of the operator defining the boundary conditions determining a self-adjoint extension of the system $A$ times the identity operator (see Theorem \ref{separab}).
Thus, self-adjoint extensions corresponding to boundary conditions with a different structure will define non-separable dynamics, and separable states will evolve into non-separable ones.    We will call such source of entangled states `boundary driven entanglement'.

It will be illustrated using a simple example how by choosing a non-trivial tensor product extension of a given self-adjoint extension of the system $A$, we obtain non-separable dynamics (see \S \ref{simple}).   Even more, we will show how by modifying the chosen self-adjoint extension, we can generate entangled states not only between the auxiliary system $A$ and system $B$, but even within system $B$ itself (as  long as it is a composite system itself).

Such instances will be discussed first by using a toy example consisting of the free particle  moving on the half-line as auxiliary system and a two-level system as a bulk system.    In this particular instance it will be shown that the ground state of the half-line (actually its only eigenstate) becomes entangled with the eigenstates of the bulk system and how such entangled state can be driven by modifying the boundary conditions compatible with such scenario.   Finally, we will discuss a ``quantum compass'', i.e. a planar rotor possessing a spin 1/2 system sitting inside it.   Now, two families of non-trivial boundary conditions for such system, extending in a non-trivial way quasi-periodic boundary conditions for the planar rotor \cite{As83}, will be considered and their spectral properties will be discussed (see \S \ref{quantum_rotor}).

%%%%%%%%%%%%%%%%%%%%%%%%%%%%%%%%%%%%%%%%%%%%%%%%%%%%%%%%%%%%%%%%%%%%%%%%%%%%%%%%%%%%%%%%%%%%%%%%%%%%%%%%%%%%%%%%%%%%%%%%%%%%
%%%%%%%%%%%%%%%%%%%%%%%%%%%%%%%%%%%%%%%%%%%%%%%%%%%%%%%%%%%%%%%%%%%%%%%%%%%%%%%%%%%%%%%%%%%%%%%%%%%%%%%%%%%%%%%%%%%%%%%%%%%%
\section{Boundary conditions and self-adjoint extensions}\label{salaplacian}

We start reviewing briefly the most salient aspects of the relation between self-adjoint extensions and boundary conditions by using the Laplace-Beltrami operator as an illuminating example.

Given a symmetric operator $T$ on a Hilbert space $\mathcal{H}$, this is, the operator $T$ has dense domain $\mathcal{D}_0 \subset \mathcal{H}$ and $T \subset T^\dagger$, we may use von Neumann's theorem \cite{Ne31} to describe all its self-adjoint extensions, if any (see for instance  \cite{Re75} for an exhaustive account of the theory).  Namely, we compute first its deficiency spaces $\mathcal{N}_\pm := \ker (T^\dagger \mp i \mathbb{I}) = \mathrm{Ran} (T \pm i\mathbb{I} )^\perp$.  Then there is a one-to-one correspondence between self-adjoint extensions of $T$  and unitary operators $K \colon \mathcal{N}_+ \to \mathcal{N}_-$. Von Neumann's theorem establishes that to any such unitary operator $K$ one can associate the self-adjoint operator $T_K$ with domain
\begin{equation}\label{domK}
\mathcal{D}_K = \mathcal{D}_0 \oplus (\mathbb{I} + K) \mathcal{N}_+ 
\end{equation}
and defined by
\begin{equation}\label{TK}
T_K(\Phi_0 \oplus (\mathbb{I} + K) \xi_+) = T\Phi_0 \oplus \, i(\mathbb{I} - K) \xi_+ , \qquad \forall \Phi_0 \in \mathcal{D}_0,\,  \xi_+\in \mathcal{N}_+  \,.
\end{equation}

In many occasions the operator $T$ is a differential operator on a manifold $\Omega$ with non-empty boundary $\partial \Omega$.    Let us consider, as an illustrative situation, a free particle moving on a curved manifold $\Omega$ with Riemannian metric $\eta$.   In such case, the Hamiltonian describing the geodesic motion is the negative Laplace-Beltrami operator $-\Delta_\eta$ (this is, we are in the situation of eq. \eqref{free} with $V \equiv 0$), that in local coordinates $x^i$, $i = 1, \ldots, n$, $n = \dim \Omega$, takes the explicit form:
\begin{equation}\label{laplacian}
\Delta_\eta = \frac{1}{\sqrt{|\eta |}} \frac{\partial}{\partial x^i} \eta^{ij} \sqrt{|\eta |}\frac{\partial}{\partial x^j} ,  
\end{equation}
with $|\eta| = \det (\eta_{ij})$.     It is natural to start by defining this operator in the domain $C_c^\infty(\mathrm{Int}(\Omega))$, i.e., in the set of complex-valued functions with compact support contained in the interior of $\Omega$, which is a dense subspace of the Hilbert space $\mathcal{H} = \L^2 (\Omega)$, the space of square integrable functions with respect to the Riemannian volume defined by $\eta$.    A simple integration by parts leads to
\begin{equation}\label{symmetric}
\scalar{\Phi}{\Delta_\eta\Psi} = \scalar{\Delta_\eta\Phi}{\Psi} \qquad \forall \Phi, \Psi \in C_c^\infty(\mathrm{Int}(\Omega))\,. 
\end{equation}
This shows that the operator $\Delta_\eta$ defined in the previous domain is symmetric.  The minimal closed extension of the operator $\Delta_\eta$ is defined on the domain $\mathcal{D}_0 = \mathcal{H}_0^2(\Omega)$, which is the closure of $C_c^\infty(\mathrm{Int}(\Omega))$ with respect to the Sobolev norm $|| \cdot ||_{2,2}$.  The domain $\mathcal{D}_0$ is just the Sobolev space of order 2 with functions that vanish at the boundary and such that their normal derivatives vanish too.

The adjoint operator $\Delta_\eta^\dagger$ is the operator defined in the domain  $\mathcal{D}^\dagger_0 = \{ \Phi \in \L^2(\Omega) \mid \Delta_\eta \Phi \in \L^2(\Omega) \}$.    Such operator $\Delta_\eta^\dagger$ is actually the maximal extension of $\Delta_\eta$ and, certainly $\Delta_\eta \subset \Delta_\eta^\dagger$.

A general result on operators commuting with conjugations shows the existence of self-adjoint extensions for $\Delta_\eta$, hence the existence of unitary operators $K\colon \mathcal{N}_+ \to  \mathcal{N}_-$ and the applicability of Neumann's theorem. eqs. \eqref{domK}, \eqref{TK}.  

Alternatively, we may argue as follows (see for instance \cite{As05} and references therein).   Consider the restriction to the boundary $\partial \Omega$ of functions in $\mathcal{D}_0^\dagger$.  Such restrictions will be denoted by $\varphi := \Phi\mid_{\partial \Omega}$.   In the same way we define the normal derivative $\dot{\varphi} := \partial \Phi /\partial \nu \mid_{\partial \Omega}$ as the outbound normal derivative along the boundary.  We will consider that both $\varphi$, $\dot{\varphi}$ are in $\L^2(\partial \Omega)$.   Repeating the integration by parts for elements $\Phi, \Psi \in \mathcal{D}_0^\dagger$ we will obtain
\begin{equation}\label{green}
\scalar{\Phi}{\Delta_\eta\Psi}-\scalar{\Delta_\eta\Phi}{\Psi} = \scalar{\dot{\varphi}}{\psi} -\scalar{\varphi}{\dot{\psi}}\,.
\end{equation}
The inner product in the r.h.s. of  the expression above is the one defined in $\L^2(\partial \Omega)$, namely, $\langle \varphi, \psi \rangle = \int_{\partial\Omega} \bar{\varphi}(x) \psi(x) d\mu_{\partial \Omega}(x)\,,$ where $\mu_{\partial\Omega}$ is the measure associated to the Riemannian metric induced at the boundary $\partial \Omega$ by $\eta$.    

Clearly, self-adjoint extensions of $\Delta_\eta$ will be determined by maximal subspaces of functions $\Phi$ in $\mathcal{D}_0^\dagger$ such that the bilinear form given by the r.h.s. of eq. \eqref{green} vanishes identically for the corresponding boundary values $\varphi$ and $\dot{\varphi}$ of $\Phi$.   

Such maximally isotropic spaces $W$ of boundary values can be easily characterized by computing their Cayley transform, i.e., we consider the linear isomorphism $C \colon \L^2(\partial \Omega) \oplus \L^2(\partial \Omega) \to \L^2(\partial \Omega) \oplus \L^2(\partial \Omega)$ defined by 
$$
C(\varphi, \dot{\varphi}) = \frac{1}{\sqrt{2}}(\varphi - i\dot{\varphi},  \varphi + i\dot{\varphi})  .
$$ 
The Cayley transform $C$ maps a maximally isotropic subspace $W$ onto the graph of a unitary operator $U \colon \L^2(\partial \Omega) \to  \L^2(\partial \Omega)$.  More explicitly $(\varphi, \dot{\varphi}) \in W$ iff there exists $U \in \mathcal{U}( \L^2(\partial \Omega))$ such that \cite{As05}:
\begin{equation}\label{asorey}
\varphi-i\dot{\varphi}=U(\varphi+i\dot{\varphi})\;.
\end{equation}
In this sense the space of self-adjoint extensions of the Laplace-Beltrami operator can be naturally identified with the unitary group of the Hilbert space of square integrable functions at the boundary of $\Omega$ and eq. \eqref{asorey} provides the explicit description of the corresponding domains. Unfortunately this description is complete only for one-dimensional Riemannian manifolds. Nevertheless, under some conditions on the unitary operator $U$, one can still characterize a wide class of self-adjoint extensions of the Laplace-Beltrami operator in arbitrary dimensions in terms of boundary conditions of  the form \eqref{asorey}, cf. \cite{Ib13} for a more detailed discussion. We will make an extensive use of this characterization in what follows.

%%%%%%%%%%%%%%%%%%%%%%%%%%%%%%%%%%%%%%%%%%%%%%%%%%%%%%%%%%%%%%%%%%%%%%%%%%%%%%%%%%%%%%%%%%%%%%%%%%%%%%%%%%%%%%%%%%%%%%%%%%%%
%%%%%%%%%%%%%%%%%%%%%%%%%%%%%%%%%%%%%%%%%%%%%%%%%%%%%%%%%%%%%%%%%%%%%%%%%%%%%%%%%%%%%%%%%%%%%%%%%%%%%%%%%%%%%%%%%%%%%%%%%%%%
\section{Self-adjoint extensions of symmetric bipartite systems}\label{bipartite_extensions}

Let us consider the case of a bipartite system $A\times B$ such that one of its subsystems is described by a symmetric operator. In particular we consider system $A$ to be defined as in the previous section by minus the Laplace-Beltrami operator on a Riemannian manifold $(\Omega_A,\eta_A)$, i.e., $A$ describes a free system on a manifold with boundary. System $B$ is defined by a self-adjoint operator $H_B$ on a Hilbert space $\H_B$ with dense domain $\operatorname{dom}(H_B)=\D_B$. The Hilbert space $\mathcal{H}_{AB}$ of pure states of the composite system is
$$
\mathcal{H}_{AB} := \H_A\cotimes\H_B=\L^2(\Omega_A)\cotimes\H_B \, ,
$$ 
that can be identified naturally with $\L^2(\Omega_A;\H_B)$\,.  Hence, pure states will be considered as square integrable maps $\Phi:\Omega_A\to\H_B$ with inner product
	\begin{equation}\label{inner product}
		\scalar{\Phi}{\Psi}_{AB}=\int_{\Omega_A}\scalar{\Phi(x)}{\Psi(x)}_{\H_B}\d\mu_\eta(x).
	\end{equation}
In what follows we will use the latter identification when appropriate.
The Hamiltonian operator of the composite system that we will consider is $H =-\Delta_\eta\otimes\mathbb{I}+\mathbb{I}\otimes H_B$, acting on states $\Phi$ as
	\begin{equation}\label{operatoraction}
		H\Phi = -\Delta_{\eta_A}\Phi + H_B\cdot\Phi,
	\end{equation}
with $(H_B\cdot\Phi)(x) = H_B(\Phi(x))$, $x\in\Omega_A$. 

The natural symmetric domain $\D_0$ of the  operator $H$ is now $\D_0=\D_{A0}\otimes\D_B$, where we are now borrowing the notation $\mathcal{D}_{A0}$ from section \ref{salaplacian} to denote the minimal closed extension of the Laplace-Beltrami operator defined on $\Omega_A$.   

Again $\mathcal{D}_{0}$ can be identified in a natural way with $\overline{\mathcal{C}_0^\infty(\Omega_A;\H_B)}^{\norm{\cdot}_{2,2}}$, where the completion is taken with respect to the Sobolev norm of order 2.  Notice that we can not consider the completion $\cotimes$ in the definition of $\mathcal{D}_0$ before because $\D_{A0}$ and $\D_B$ being dense, it would result that $\D_{A0}\cotimes\D_B = \H_A\hat{\otimes}\H_B = \mathcal{H}_{AB}$\,, but the operator $H$ is not bounded.  

The maximal extension of the operator $H$ is given by $\D_{A0}^\dagger\otimes\D_B$ using the notation of  Section \ref{salaplacian} again (notice that $H_B$ is self-adjoint already, then $\mathcal{D}_B^\dagger = \mathcal{D}_B$).   Computing the self-adjoint extensions of $H$ is best done by using its boundary data structure (i.e., Green's formula) like in the second part of Section \ref{salaplacian}.  In fact, integrating by parts we get the analogue of eq. \eqref{green}:
	\begin{equation}\label{green bipartite}
		\scalar{\Phi}{-\Delta_{\eta_A}\Psi+H_B\cdot\Psi}-\scalar{-\Delta_{\eta_A}\Phi+H_B\cdot\Phi}{\Psi}=\scalar{\varphi}{\dot{\psi}}-\scalar{\dot{\varphi}}{\psi}\,,
	\end{equation}
where the inner product at the boundary appearing in the r.h.s. of the previous equation is given simply by
	\begin{equation}\label{inner product boundary}
		\scalar{\varphi}{\psi}=\int_{\partial\Omega_A}\scalar{\varphi(x)}{\psi(x)}_{\H_B}\d\mu_{\partial\eta_A}(x)
	\end{equation}
and $\varphi$, $\dot{\varphi}$ are defined as before. Then $\varphi$, $\dot{\varphi}$  can be identified with functions on $\partial \Omega$ with values in $\mathcal{H}_B$ and the space of boundary data is now $\L^2(\partial\Omega_A;\H_B)\simeq\L^2(\partial\Omega_A)\hat{\otimes} \H_B $. 

Repeating the argument leading to eq. \eqref{asorey}, we obtain that the space of self-adjoint extensions of $H$, i.e., the space of maximally isotropic, closed subspaces of the bilinear boundary form defined by the r.h.s. of  eq. \eqref{green bipartite}, is parametrized by unitary operators $U\in\mathcal{U}(\L^2(\partial\Omega_A)\hat{\otimes}\H_B)$.  Thus, given an unitary operator $U\colon \L^2(\partial\Omega_A)\hat{\otimes}\H_B \to \L^2(\partial\Omega_A)\hat{\otimes}\H_B$, the domain $\D_U$ of the corresponding self-adjoint extension will consist of all functions $\Phi\in\D_{A0}^\dagger\otimes\D_B$ such that 
	\begin{equation}\label{asorey bipartite}
		\varphi-i\dot{\varphi}=U(\varphi+i\dot{\varphi})\;,\quad \varphi\;, \dot{\varphi}\in\L^2(\partial\Omega_A;\H_B)\;.
	\end{equation}

A similar result, but in a much more general situation, can be obtained certainly by using von Neumann's Theorem (now $\mathcal{H}_A$ and $\mathcal{H}_B$ are general complex separable Hilbert spaces and $H_A$, $H_B$ operators on them).

\begin{theorem}\label{deficiency}
Let $H_A$ be a densely defined, symmetric operator on the Hilbert space $\mathcal{H}_A$ and $H_B$ a bounded, self-adjoint operator on a Hilbert space $\mathcal{H}_B$ with discrete spectrum, then the deficiency spaces $\mathcal{N}_\pm$ of the symmetric operator $H = H_A\otimes \mathbb{I} + \mathbb{I}\otimes H_B$ are isomorphic to
$\N_{A,\pm}\cotimes\H_B$.
\end{theorem}

\begin{proof}  Let us assume for simplicity that $H_B$ has non degenerate discrete spectrum $\lambda_n$ with eigenvectors $\rho_n$, $H_B \rho_n = \lambda_n \rho_n$.  Then, the normalized eigenvectors $\rho_n$ define an orthonormal basis for $\H_B$ and any vector $\Phi\in\L^2(\partial\Omega_A)\hat{\otimes}\H_B$ has a unique representation as
	\begin{equation}\label{representation AB}
		\Phi=\sum_n\Phi_n\otimes\rho_n,\quad\Phi_n\in\L^2(\partial\Omega_A).
	\end{equation}
Hence, we get for vectors $\Phi^\pm \in \mathcal{N}_\pm = \ker (H^\dagger \mp i\mathbb{I})$:
	\begin{eqnarray*}
		(H^\dagger\mp iI)\Phi^\pm	& = & (-\Delta^\dagger_{\eta_A}\otimes\mathbb{I}+\mathbb{I}\otimes H_B) \left( \sum_n\Phi^\pm_n\otimes\rho_n \right)  \mp i I\Phi^\pm  =  \nonumber \\
							& = & \sum_n (-\Delta^\dagger_{\eta_A}\Phi^\pm_n+\lambda_n\Phi^\pm_n\mp i \Phi^\pm_n)\otimes\rho_n=0
	\end{eqnarray*}
which implies
	\begin{equation}
		-\Delta^\dagger_{\eta_A}\Phi^\pm_n\mp(i\mp\lambda_n)\Phi^\pm_n=0.
	\end{equation}
Thus $\Phi^\pm_n$ must belong to the generalized deficiency spaces 
\begin{align*}
	\N_{A,z_n} &= \{ \Phi^+\in \D_{A0}^\dagger \mid -\Delta_{\eta A}^\dagger \Phi^+ = z_n \Phi^+ \},\\
	\N_{A,\bar{z}_n } &= \{ \Phi^-\in \D_{A0}^\dagger \mid -\Delta_{\eta A}^\dagger \Phi^- = \bar{z}_n \Phi^- \} , 
\end{align*}
with $z_n = (-\lambda_n + i)$.    However, all generalized deficiency spaces of the form $\N_{A,z}$ with $\operatorname{Im} z>0$ are isomorphic, that is, $\operatorname{dim}\N_{A,z}$ is constant in the upper complex half-plane (similarly, if $\operatorname{Im} z<0$, then $\operatorname{dim}\N_{A,z}$ is constant in the lower half-plane) and hence $\N_{A,\pm} = \N_{A,\pm i}$ is isomorphic to $\N_{A,\pm(i\mp\lambda_n)}$ (see for instance \cite{Ak60}).  Let us denote a choice for such isomorphism by $\alpha^\pm_n:\N_{A,\pm(i\mp\lambda_n)}\to\N_{A,\pm}$. 

We have shown that the deficiency spaces $\N_{\pm}$ of the operator $H$ consist of vectors of the form $\sum_n\Phi^\pm_n\otimes\rho_n$ with $\Phi^\pm_n\in\N_{A,\pm(i\mp\lambda_n)}$. The isomorphism $\alpha^\pm: \N_{\pm}\to\N_{A,\pm}\cotimes\H_B$ defined by 
	\begin{equation}
		\alpha^\pm (\sum\Phi^\pm_n\otimes\rho_n)=\sum\alpha^\pm_n(\Phi^\pm_n)\otimes\rho_n
	\end{equation}
provides an explicit identification of $\N_\pm$ with $\N_{A,\pm}\cotimes\H_B$. \end{proof}

The previous argument generalizes to the case of a general self-adjoint operator $H_B$ by a judiciously use of its spectral representation.

Notice that as a consequence of the previous theorem, the space of self-adjoint extensions of the composite system defined by the Hamiltonian $H$ is given by the space of unitary operators $K \colon \mathcal{N}_{A,+}\cotimes\H_B \to \mathcal{N}_{A,-}\cotimes\H_B$\,, which is much larger that the space of self-adjoint extensions of the system $A$ alone.  In particular, the self-adjoint extensions defined by unitary operators of the form $K_+ \otimes \mathbb{I}$ are in one-to-one correspondence with self-adjoint extensions of the system $A$ alone, $K_A \colon \mathcal{N}_{A,+} \to \mathcal{N}_{A,-}$. 

%%%%%%%%%%%%%%%%%%%%%%%%%%%%%%%%%%%%%%%%%%%%%%%%%%%%%%%%%%%%%%%%%%%%%%%%%%%%%%%%%%%%%%%%%%%%%%%%%%%%%%%%%%%%%%%%%%%%%%%%%%%%
%%%%%%%%%%%%%%%%%%%%%%%%%%%%%%%%%%%%%%%%%%%%%%%%%%%%%%%%%%%%%%%%%%%%%%%%%%%%%%%%%%%%%%%%%%%%%%%%%%%%%%%%%%%%%%%%%%%%%%%%%%%%
\section{Separable dynamics and separable extensions}\label{separable dynamics}

It is clear that if we have two complete quantum systems $A$ and $B$ with Hilbert spaces of state vectors $\H_A$ and $\H_B$, and Hamiltonian operators $H_A$ and $H_B$ respectively, then the bipartite system with Hilbert space $\H = \H_A \hat{\otimes} \H_B$ and total Hamiltonian $H = H_A\otimes\mathbb{I}+\mathbb{I}\otimes H_B$ induces a unitary flow 
\begin{align*}
U_t=e^{itH} &= e^{it(H_A\otimes \mathbb{I}+\mathbb{I}\otimes H_B)}=  e^{it(H_A\otimes \mathbb{I})} e^{it(\mathbb{I}\otimes H_B)} \\ &= e^{itH_A}\otimes e^{itH_B}=U_t^A\otimes U_t ^B\,,
\end{align*}
where $U^A_t$, $U^B_t$ denote the individual unitary flows of the subsystems $A$ and $B$. Then we may call a one-parameter family of unitary operators $U_t$ on $\H=\H_A\hat{\otimes}\H_B$ separable if there exist two one-parameter families of unitary operators $U_t^A$ and $U_t^B$ on $\H_A$ and $\H_B$ respectively such that
	\begin{equation}
		U_t=U_t^A\otimes U_t ^B\;.
	\end{equation}
Notice that $U_t$ is separable if and only if $U_t\Phi$ is separable for any separable state $\Phi=\Phi_A\otimes\Phi_B$ for any $t$.  Even more, it is immediate to check that separable dynamics do not change the Schmidt index of a given state in $\H_A\hat{\otimes}\H_B$. 

Now, if we are given a system $H$ on $\H_A\otimes\H_B$ which is obtained by means of a self-adjoint extension of the product of a symmetric operator on $\H_A$ and a self-adjoint operator on $\H_B$, can we determine when are we going to obtain separable dynamics?  In other words, if $U\in\mathcal{U}(\L^2(\partial\Omega_A)\cotimes\H_B)$ is the unitary operator defining the self-adjoint extension,cf. eq. \eqref{asorey bipartite},  under what conditions will it characterize separable separable dynamics? 

We will solve first the spectral problem for the self-adjoint extension of $H$ defined by the boundary condition $U = U_A \otimes \mathbb{I}$.    We will assume for simplicity in what follows that the spectrum of $H_B$ is discrete and non-degenerate.  We denote the eigenvalues and eigenvectors of $H_B$ by $\lambda_k^B$ and $\rho_k^B$ respectively, $H_B \rho_k^B = \lambda_k^B \rho_k^B$, $k = 1,2,\ldots$.  An arbitrary function $\Phi \in \L^2(\Omega_A; \H_B)$ can be written uniquely as
$$ \Phi = \sum_{k =1}^\infty  \Phi_k^A \otimes \rho_k^B ,$$
hence 
$$\varphi = \sum_{k=1}^\infty \varphi_k^A \otimes \rho_k^B, \quad \dot{\varphi} = \sum_{k=1}^\infty \dot{\varphi}_k^A \otimes \rho_k^B ,$$
with $\varphi_k^A = \Phi_k^A\mid_{\partial \Omega_A}$ and $\dot{\varphi}_k^A = d\Phi_k^A/d\nu \mid_{\partial \Omega_A}$, $k = 1,2,\ldots$.  If 
$U = \sum_{s= 1}^N U_s^A \otimes U_s^B$ is a unitary operator acting on $\L^2(\partial \Omega_A)\hat{\otimes} \H_B$, we have $ U\varphi = \sum_{k=1}^\infty\sum_{s=1}^N U_s^A\varphi_k \otimes U_s^B \rho_k^B$ and so on.  In the particular instance of $U = U_A\otimes \mathbb{I}$ we get
$$ U\varphi = \sum_{k=1}^\infty U_A\varphi_k^A \otimes \rho_k^B $$
and similarly for $\dot{\varphi}$.  If we denote by $H_U$ the self-adjoint extension defined by $U$, the spectral problem $H_U \Phi = E \Phi$ becomes, after some trivial computations, the family of spectral problems:
\begin{equation}\label{spectralk}
H_A^\dagger \Phi_k^A + \lambda_k^B \Phi_k^A = E \Phi_k^A, \quad k = 1,2, \dots
\end{equation}
and the boundary conditions defined by $U$ become the family of boundary conditions:
$$
 \varphi_k^A - i\dot{\varphi}_k^A = U_A(\varphi_k^A + i\dot{\varphi}_k^{A}), \quad k = 1,2, \ldots .
$$
Thus, for each $k$ we have to solve the problem:

\begin{equation}\label{spectralsep}
		\begin{cases}
			&H_A^\dagger \Phi^A = ( E- \lambda_k^B) \Phi^A \\
			&\varphi^A - i\dot{\varphi}^A =  U_A(\varphi^A + i\dot{\varphi}^A) .
		\end{cases}
	\end{equation}
Notice that if we denote by $\Psi_l^A$ the eigenfunctions of the self-adjoint extension of the operator $H_A$ defined by $U_A$, this is, with the boundary conditions given in eq. \eqref{spectralsep}, we will have:
$$ H_A^\dagger \Psi_l ^A = \lambda_l^A \Psi_l^A, \quad \psi_l^A - i\dot{\psi}_l^A =  U_A(\psi_l^A + i\dot{\psi}_l^A), \quad l = 1,2, \ldots $$

We will also assume, for simplicity, that the spectrum of the extension of $H_A$ defined by $U_A$ is discrete (this supposes no loss of generality for our purposes).  In particular in dimension 1 the spectrum of any self-adjoint extension of the Laplace operator is discrete \cite{We80}. In general using the spectral theorem \cite{Ak60}, one can adapt this construction to the general case.   

We will denote in what follows by $H_{U_A}$ the self-adjoint extensions of $H_A$ determined by the unitary $U_A$\,.
We finally conclude that the spectrum of $H_U$ is given by
$$ E = \lambda_l^A + \lambda_k^B, \qquad k,l = 1,2, \ldots , $$
with eigenvectors $\Psi_l^A \otimes \rho_k^B$.  Now, if $\Phi \in \L^2 (\Omega_A; \H_B)$, we have that
$$ \Phi = \sum_{k,l} c_{lk} \Psi_l^A \otimes \rho_k^B$$
with $c_{lk} = \langle \Psi_l^A \otimes \rho_k^B, \Phi \rangle$, and if $\Phi$ is separable, $\Phi = \Phi_A \otimes \Phi_B$, we obtain:
$$c_{lk} = \langle \Psi_l^A, \Phi_A \rangle \langle \rho_k^B, \Phi_B \rangle = a_l b_k ,$$
with $\Phi_A = \sum_l a_l \Psi_l^A$ and $\Phi_B = \sum_k b_k \rho_k^B$ respectively.   Consequently,
\begin{eqnarray*}
 e^{itH_U}\Phi_A \otimes \Phi_B &=& \sum_{k,l} a_lb_k e^{itH_U}(\Psi_l^A \otimes \rho_k^B)\\ 
 &=& \sum_{k,l} a_lb_k e^{i(\lambda_l^A + \lambda_k^B)}(\Psi_l^A \otimes \rho_k^B) \\ &=& \left(e^{itH_{U_A}}\otimes e^{itH_B}\right)(\Phi_A \otimes \Phi_B)
 \end{eqnarray*}
 which shows that the self-adjoint extension defined by the unitary matrix $U = U_A \otimes \mathbb{I}$ is separable as it was easy to presume.

Let us discuss now boundary conditions of the simple form:
\begin{equation}\label{UU}
 U = U_A \otimes U_B ,
 \end{equation}
 with $U_A \in \mathcal{U}(\L^2(\partial \Omega_A))$ and $U_B \in \mathcal{U}(\H_B)$,
i.e., decomposable elements in the unitary group $\mathcal{U}(\L^2(\partial \Omega_A)\hat{\otimes}\H_B)$. 
We may even consider for simplicity that the unitary $U_B$ defines a symmetry of the quantum system $H_B$, this is $[H_B, U_B] = 0$.   In this case and in contrary to a simple guess, the dynamics defined by $U$ of the form in eq. \eqref{UU} is non-separable if $U_B \neq \mathbb{I}$\,.

The proof of this fact is as follows.   Because $U_B$ is a unitary operator, it can be diagonalized and the Hilbert space $\H_B$ decomposed as $\H_B = \bigoplus_{s= 1}^\infty W_s$, with $W_s$ orthogonal $U_B$-invariant subspaces such that:
$$ U_B \Phi_s^B = e^{i\nu_s} \Phi_s^B, \qquad \forall \Phi_s^B \in W_s .$$

Now, because $U_B$ commutes with $H_B$, $H_B$ will leave the subspaces $W_s$ invariant too, and we will denote by $H_{B,s}$ the restriction of $H_B$ to $W_s$, $s= 1, 2, \ldots$.  Moreover, we have that
$$ \H = \L^2 (\Omega_A;\H_B) = \L^2 (\Omega_A;\bigoplus_{s= 1}^\infty W_s) = \bigoplus_{s= 1}^\infty \L^2 (\Omega_A;W_s) .$$
The operator $H_A^\dagger$ leaves invariant the subspaces $\L^2 (\Omega_A;W_s)$ for all $s$.  Hence, the spectral problem $H_U\Phi = E\Phi$ with boundary conditions defined by eq. \eqref{UU} is equivalent to the solution of the family of spectral problems
\begin{equation}\label{spectralsubs}
		\begin{cases}
			& (H_A^\dagger\otimes \mathbb{I} + \mathbb{I}\otimes H_B) \Phi_s = E_s \Phi_s \\
			&\varphi_s - i\dot{\varphi}_s =  (U_A\otimes e^{i\nu_s})(\varphi_s + i\dot{\varphi}_s), \qquad s = 1,2,\ldots ,
		\end{cases}
	\end{equation}
because $\Phi_s \in \L^2(\Omega_A; W_s) \equiv \L^2(\Omega_A)\hat{\otimes} W_s$; $\varphi_s, \dot{\varphi}_s \in \L^2(\partial\Omega_A)\hat{\otimes}\H_B$, and $(U_A\otimes U_B) (\varphi_s+i\dot{\varphi}_s) = (U_A\otimes e^{i\nu_s})(\varphi_s+i\dot{\varphi}_s)$.  But the boundary conditions in eq. \eqref{spectralsubs} are the same as
$$ \varphi_s - i\dot{\varphi}_s =  \widetilde{U}_{A,s}(\varphi_s + i\dot{\varphi}_s), \qquad s = 1,2, \ldots $$
with $\widetilde{U}_{A,s} = e^{i\nu_s}U_A$, and in consequence the self-adjoint extension defined in $\L^2(\Omega_A;W_s)$ of the Hamiltonian $H_s = H_A\otimes \mathbb{I} + \mathbb{I} \otimes H_{B,s}$ by the boundary conditions $\widetilde{U}_{A,s}\otimes \mathbb{I}$ is separable because of the preceding arguments.

In consequence we have obtained that:
$$ e^{itH_s} = e^{itH_{\tilde{U}_{A,s}}} \otimes e^{it H_{B,s}}  .$$
Finally, if $\Phi = \Phi_A \otimes \Phi_B$ is a separable state in $\H$, we can write it as:
$$ \Phi = \sum_{s= 1}^\infty \sum_{i =1}^{N_s} \Phi_{s,i}^A \otimes \Phi_{s,i}^B $$
with $\Phi_{s.i}^B \in W_s$.  Then we have that
\begin{align}\label{dynamicss}
e^{itH_U} \Phi =& \sum_{s= 1}^\infty \sum_{i =1}^{N_s} e^{itH_U}  \Phi_{s,i}^A \otimes \Phi_{s,i}^B \\=& \sum_{s= 1}^\infty \sum_{i =1}^{N_s} e^{itH_{\tilde{U}_{A,s}}}  \Phi_{s,i}^A \otimes e^{itH_{B,s}} \Phi_{s,i}^B \;.\notag
\end{align}
But now the factor $e^{i\nu_s}$ is different for each $s$, hence the eigenvalues and eigenvectors of the spectral problem in $\L^2(\Omega_A; W_s)$, $s = 1,2, \ldots$
 \begin{equation}\label{spectrals}
		\begin{cases}
			& (H_A^\dagger\otimes \mathbb{I} + \mathbb{I}\otimes H_{B,s}) \Phi_s = E_s \Phi_s \\
			&\varphi_s - i\dot{\varphi}_s =  e^{i\nu_s} U_A (\varphi_s + i\dot{\varphi}_s), \qquad s = 1,2,\ldots ,
		\end{cases}
\end{equation}
are different for each $s$. Therefore the extension $H_{\tilde{U}_{A,s}}$ is different for each $s$ and we cannot factorize it out of the sum in the last term on the r.h.s. of eq. \eqref{dynamicss}. Thus we conclude that if $U_B \neq \mathbb{I}$ the dynamics $H_U$ is non-separable. Notice that the case $\nu_s=\nu$, $s=1,2\dots$ is equivalent to $U=e^{i\nu}U_A\otimes \mathbb{I}$\,.

We can prove the following theorem.

\begin{theorem}\label{separab}  Let $H_A$ be a densely defined symmetric operator on $\mathcal{L}^2(\Omega)$ and $H_B$ a self-adjoint operator on the Hilbert space $\mathcal{H}_B$.  Let $U_A$ be a unitary operator on $\mathcal{L}^2(\partial \Omega)$ such that  self-adjoint extensions of the operator $H_A$ defined by $e^{i\alpha}U$ have discrete spectrum for all $0\leq \alpha\leq 2\pi$.Let $H_B$ have discrete spectrum.  Then, the dynamics $H_U$ on the product Hilbert space $\L^2(\Omega_A)\cotimes\H_B$ defined by the unitary operator $U = U_A \otimes U_B\in\mathcal{U}(\L^2(\partial\Omega_A)\cotimes\H_B)$ is separable iff $U_B = \mathbb{I}$.
\end{theorem}

\begin{proof}  Let us assume that the dynamics defined by $H_U$ is separable. Then:
$$e^{itH_U} = e^{it\tilde{H}_A} \otimes e^{it\tilde{H}_B} ,$$
where, in principle, neither $\tilde{H}_A$ nor $\tilde{H}_B$ have to agree with $H_{U_A}$ nor $H_B$ respectively.  However we have that
$$H_U = \tilde{H}_A\otimes \mathbb{I} + \mathbb{I} \otimes \tilde{H}_B$$
with $\tilde{H}_A$ and $\tilde{H}_B$ self-adjoint operators.   It is also clear that the one-parameter group of unitary operators
$$ V_t = e^{it\tilde{H}_B}$$
defines a symmetry group of $H_U$,
$$ [H_U, V_t] = 0, \qquad \forall t .$$
Moreover the group $V_t$ acts unitarily on the boundary space $\L^2(\partial \Omega_A;\H_B)$.  It can be showed, cf. \cite{Ib14}, that then necessarily
$$ [U, I\otimes \tilde{H}_B ] = 0 \,.$$
Hence, $[U_B, \tilde{H}_B ] = 0$.   But now we have a self-adjoint extension defined by a unitary matrix of the form $U_A\otimes U_B$ with $[U_B, \tilde{H}_B ] = 0$ as in the discussion preceding this theorem. Then, repeating the previous arguments we will obtain that the dynamics is non-separable unless $U_B = \mathbb{I}$.
\end{proof}

%%%%%%%%%%%%%%%%%%%%%%%%%%%%%%%%%%%%%%%%%%%%%%%%%%%%%%%%%%%%%%%%%%%%%%%%%%%%%%%%%%%%%%%%%%%%%%%%%%%%%%%%%%%%%%%%%%%%%%%%%%%%%%%%%%%%%%%%%%%%%%%%%%%%%%%%%%%%%%%%%%%%%%%%%%%%%%%%%%%%%%%%%%%%%%%%%%%%%%%%%%%%%%%%%%%%%%%%%%%%%%%%%%%%%%%%%%%%%%%%%%%%%%%%

\section{A simple example: the half-line/half-spin bipartite system}\label{simple}

We will discuss now what is conceivably the simplest non-trivial example of a bipartite system of the kind considered in section \ref{bipartite_extensions}. Let the auxiliary system $A$ be a free particle moving on the half-line $\mathbb{R}^+$ ($\Omega_A = \mathbb{R}^+$, $\partial \Omega_A = \{ 0 \}$). That is the Hilbert space of the system is $\H_A = \L^2(\mathbb{R}^+,dx)$ and the dynamics of that system is governed by the free Hamiltonian $-\frac{1}{2}\frac{\d^2}{\d x^2}$. The bulk system $B$ will be a 2-level system, for instance a spin $1/2$ system whose Hilbert space is $\mathbb{C}^2$. The dynamics is given by an arbitrary $2\times2$ Hermitean matrix $H_B$. We assume that $\sigma( H_B)=\{\lambda_1>\lambda_2\}$ with eigenvectors $\rho_1,\rho_2$ respectively. The corresponding bipartite system $A\times B$ is defined in the Hilbert space $\H = \H_1\cotimes\H_2= \L^2(\mathbb{R}^+)\cotimes \mathbb{C}^2\simeq  \L^2(\mathbb{R}^+; \mathbb{C}^2)$ whose state vectors $\Phi\in\H$ will be written as
	\begin{equation}
		\Phi=\Phi_1\otimes\rho_1+\Phi_2\otimes\rho_2\simeq
		\begin{bmatrix}
			\Phi_1(x) \\ \Phi_2(x)
		\end{bmatrix}
		,\quad \Phi_a(x) \in \L^2(\mathbb{R}^+), \;a=1,2
	\end{equation}
where we have used the orthonormal basis $\{\rho_1,\rho_2\}$ to write the component vectors.

As we showed before, see Theorem \ref{deficiency}, the deficiency spaces are easy to compute and we get: $\N_{\pm}=\N_{A,\pm}\otimes\mathbb{C}^2\simeq\mathbb{C}^2$, because, as it is easy to check, $\dim \N_{A,\pm} = 1$ and therefore $\N_{A,\pm} = \mathbb{C}$. However, we work directly with boundary values which will prove to be more efficient.  Thus, given $\Phi\in\H$, the boundary values of $\Phi$ will live in $\L^2(\partial\mathbb{R}^+)\otimes\mathbb{C}^2$, in fact:
	\begin{equation*}
		\varphi:=\Phi\bigr|_{\partial\Omega_A}=\Phi_1\bigr|_{\partial\Omega_A}\otimes\rho_1+\Phi_2\bigr|_{\partial\Omega_A}\otimes\rho_2=
		\begin{bmatrix}
			\Phi_1(0) \\ \Phi_2(0)
		\end{bmatrix}=:
		\begin{bmatrix}
			\varphi_1 \\ \varphi_2
		\end{bmatrix}
	\end{equation*}
and similarly
	\begin{align*}
		\dot{\varphi}:=-\frac{\partial\Phi}{\partial x}\bigr|_{\partial\Omega_A}&=\dot{\varphi}_1\otimes\rho_1+\dot{\varphi}_2\otimes\rho_2=
		\begin{bmatrix}
			\dot{\varphi}_1 \\ \dot{\varphi}_2
		\end{bmatrix};\\
		\quad\dot{\varphi}_a&= - \frac{\partial\Phi_a}{\partial x}\Bigr|_{x = 0},\; \qquad a=1,2\;.
	\end{align*}
Finally, we combine the boundary data as:
	\begin{equation*}
		\varphi_{\pm}=\varphi \pm i\dot{\varphi}=
		\begin{bmatrix}
			\varphi_1 \pm i\dot{\varphi}_1\\ \varphi_2 \pm i\dot{\varphi}_2
		\end{bmatrix}
	\end{equation*}
and the self-adjoint extensions of $H =-\frac{\d^2}{\d x^2}\otimes\mathbb{I}+\mathbb{I}\otimes H_B $ are characterized by unitary operators $U\in\mathcal{U}( \L^2(\partial\Omega_A)\otimes\mathbb{C}^2)\simeq  U(2)$ defining the domains $\varphi_-= U \varphi_+$.

Notice that in matrix form the operator $H$ has the form
	\begin{equation}\label{H_half_line}
		H = -\frac{\d^2}{\d x^2}\otimes\mathbb{I}+\mathbb{I}\otimes H_B=
		\begin{bmatrix}
			-\frac{\d^2}{\d x^2}+\lambda_1 & 0 \\ 0 & -\frac{\d^2}{\d x^2}+\lambda_2
		\end{bmatrix}.
	\end{equation}

We recall now that the boundary data space is given by $\L^2(\partial\mathbb{R}^+)\otimes\mathbb{C}^2$. Hence according with Theorem \ref{separab}, separable dynamics will be given by unitary operators of the form $U=U_A\otimes\mathbb{I}$\,, where $U_A\colon \L^2(\partial\mathbb{R}^+)\to \L^2(\partial\mathbb{R}^+)$ and hence $U_A=e^{i\alpha}$ is just multiplication by a phase.  Incidentally we may recall that these are all the self-adjoint extensions of the system $A$ in the half-line and they correspond to boundary conditions of the form
	\begin{equation}
		\varphi_A-i\dot{\varphi_A}=e^{i\alpha}(\varphi_A+i\dot{\varphi}_A)
	\end{equation}
or equivalently
	\begin{equation}\label{local }
		\begin{cases}
			\dot{\varphi}_A=-\tan(\alpha/2)\varphi_A, & \alpha\neq\pi\\
			\varphi_A=0 , & \alpha=\pi;.
		\end{cases}
	\end{equation}
Now, because the space of self-adjoint extensions for the bipartite system is actually $U(2)$\,, as it was shown above, there are many self-adjoint extensions that will define non-separable dynamics.   Notice that because the spectrum of the Laplace operator in the half-line is not discrete, we cannot apply Theorem \ref{separab}. However, we will proceed by a direct computation of the ground state of the composite system under different self-adjoint extensions.

We will consider the particular instance of self-adjoint extensions defined by unitary matrices of the form
	\begin{equation}
		U=U_A\otimes V\;,\quad U_A \in\mathcal{U}(\L^2(\partial\mathbb{R}^+))\;,V\neq \mathbb{I}\in U(2)\;.
	\end{equation}
Despite of their form they determine non-separable dynamical evolution. In fact, among this class and because $U_A$ is just multiplication by a complex number of modulus $1$, we can just consider as the simplest, non-trivial example a matrix $V$ of the form
	\begin{equation}\label{non-separable condition}
		V=
		\begin{bmatrix}
			e^{i\alpha_1} & 0 \\ 0 & e^{i\alpha_2}
		\end{bmatrix}\;,\quad\text{with}\quad e^{i\alpha_1}\neq e^{i\alpha_2}\;,
	\end{equation}
i.e., a matrix $V$ belonging to a maximal torus inside $U(2)$.  It is also noticeable that such $V$ is the most general matrix commuting with $H_B$. Notice that if $\varphi=\varphi_1\otimes\rho_1+\varphi_2\otimes\rho_2\in \L^2(\partial\mathbb{R}^+)\otimes\H_B$, then
	\begin{equation}
		(\mathbb{I}\otimes V)\varphi =\varphi_1\otimes V\rho_1 + \rho_2\otimes V\rho_2 
			=	\begin{bmatrix}
					V_{11} \varphi_1+V_{12}\varphi_2 \\ V_{21}\varphi_1 + V_{22}\varphi_2
				\end{bmatrix}=V\cdot
				\begin{bmatrix}
					\varphi_1 \\ \varphi_2
				\end{bmatrix}
			=V\cdot\varphi\;.
	\end{equation}

To compute the point spectrum of the self-adjoint operator $H_U$ defined by the unitary operator $U = \mathbb{I}\otimes V$ is easy.  Notice that 
eq. \eqref{asorey bipartite} becomes now:
	\begin{equation}\label{diagonal condition}
		\begin{bmatrix}
			\varphi_{1} - i \dot{\varphi}_1 \\ \varphi_2 - i \dot{\varphi}_2
		\end{bmatrix} =
		\begin{bmatrix}
			e^{i\alpha_1} & 0 \\ 0 & e^{i\alpha_2}
		\end{bmatrix}
		\begin{bmatrix}
			\varphi_{1} + i \dot{\varphi}_1 \\ \varphi_2 + i \dot{\varphi}_2
		\end{bmatrix}\;,
	\end{equation}
this is, $\varphi_{a-}=e^{i\alpha_a}\varphi_{a+}$, $a=1,2$ or, if both $\alpha_1,\alpha_2\neq\pi$,
	\begin{equation}\label{local non-separable}
		\dot{\varphi}_a=-\tan(\alpha_a/2)\varphi_a\;,\quad a=1,2\;.
	\end{equation}
Then, the eigenvalue problem $H_U\Phi=E\Phi$ becomes
	\begin{align*}
		H_U\Phi	&=-\frac{\d^2}{\d x^2}\Phi_1\otimes\rho_1-\frac{\d^2}{\d x^2}\Phi_2\otimes\rho_2 +\Phi_1\otimes H_B\rho_1+\Phi_2\otimes H_B \rho_2\\
				&=(-\frac{\d^2}{\d x^2}+\lambda_1)\Phi_1\otimes\rho_1+(-\frac{\d^2}{\d x^2}+\lambda_2)\Phi_2\otimes\rho_2\;.
	\end{align*}
This eigenvalue problem is equivalent to
	\begin{equation}
		\begin{matrix}
			\displaystyle{-\frac{\d^2}{\d x^2}\Phi_a=(E-\lambda_a)\Phi_a} \\
			\displaystyle{\dot{\varphi}_a=-\tan(\alpha_a/2)\varphi_a }
		\end{matrix}\biggr\} \qquad a=1,2\;.
	\end{equation}
We may start solving:
\begin{equation}
	-\frac{\d^2\Phi_1}{\d x^2}=(E-\lambda_1)\Phi_1\;;\quad\dot{\varphi_1}=-\tan(\alpha_1/2)\varphi_1\;.
\end{equation}
We see immediately that if $\lambda_1 \leq E$, the solutions to this problem are not in $\L^2(\mathbb{R^+})$, thus $\lambda_1>E$ and the corresponding eigenfunction is $\Phi_1(x)=C_1 e^{-\sqrt{\lambda_1-E}\, x}$.  Moreover $\dot{\varphi_1} = -\frac{\d\Phi_1}{\d x}\bigr|_{x=0}=C_1\sqrt{\lambda_1-E}$, hence $\sqrt{\lambda_1-E}=\tan(\alpha_1/2)$ or
$$ E=\lambda_1-\tan^2(\alpha_1/2)\;.$$
Notice that $E$ is the unique discrete eigenvalue of the operator and that the rest of the spectrum is continuous.  

We can proceed similarly for the other component ($a=2$) finding again that $E=\lambda_2-\tan^2(\alpha_2/2)$ if $E < \lambda_2$.   In consequence, if $E < \lambda_2$ we obtain the compatibility condition (recall that $\lambda_1 > \lambda_2$)
	\begin{equation}\label{compatibility}
		\tan^2(\alpha_1/2)-\tan^2(\alpha_2/2)=\lambda_1-\lambda_2>0\;,
	\end{equation}
that must be satisfied for the existence of an eigenvector with eigenvalue $\lambda_1 > \lambda_2 > E$.    Figure \ref{phasespace} shows the space of self-adjoint extensions $(\alpha_1,\alpha_2)$ with non-degenerate ground state $E$ for various values of the spectral gap $\sigma := \lambda_1 - \lambda_2$ of the bulk system.

If $\lambda_2 \leq E < \lambda_1$, $E$ is an eigenvalue again, but this time the eigenvector is going to have only the $a = 1$ component.
We want to stress that the compatibility condition eq. \eqref{compatibility} is only necessary for the existence of a non-void point spectrum.  If it is not satisfied, then the problem has no point spectrum. Nevertheless, the Hamiltonian is self-adjoint even if there is no point spectrum.

%%%%%%%%%%%%%%%%%%%%%%%%%%%%%%%%%%%%%%%%%%%%%%%%%%%%%%%%%%%%%%%%%%%%

\begin{figure}[ht]
	\centering
	\includegraphics[height=7cm]{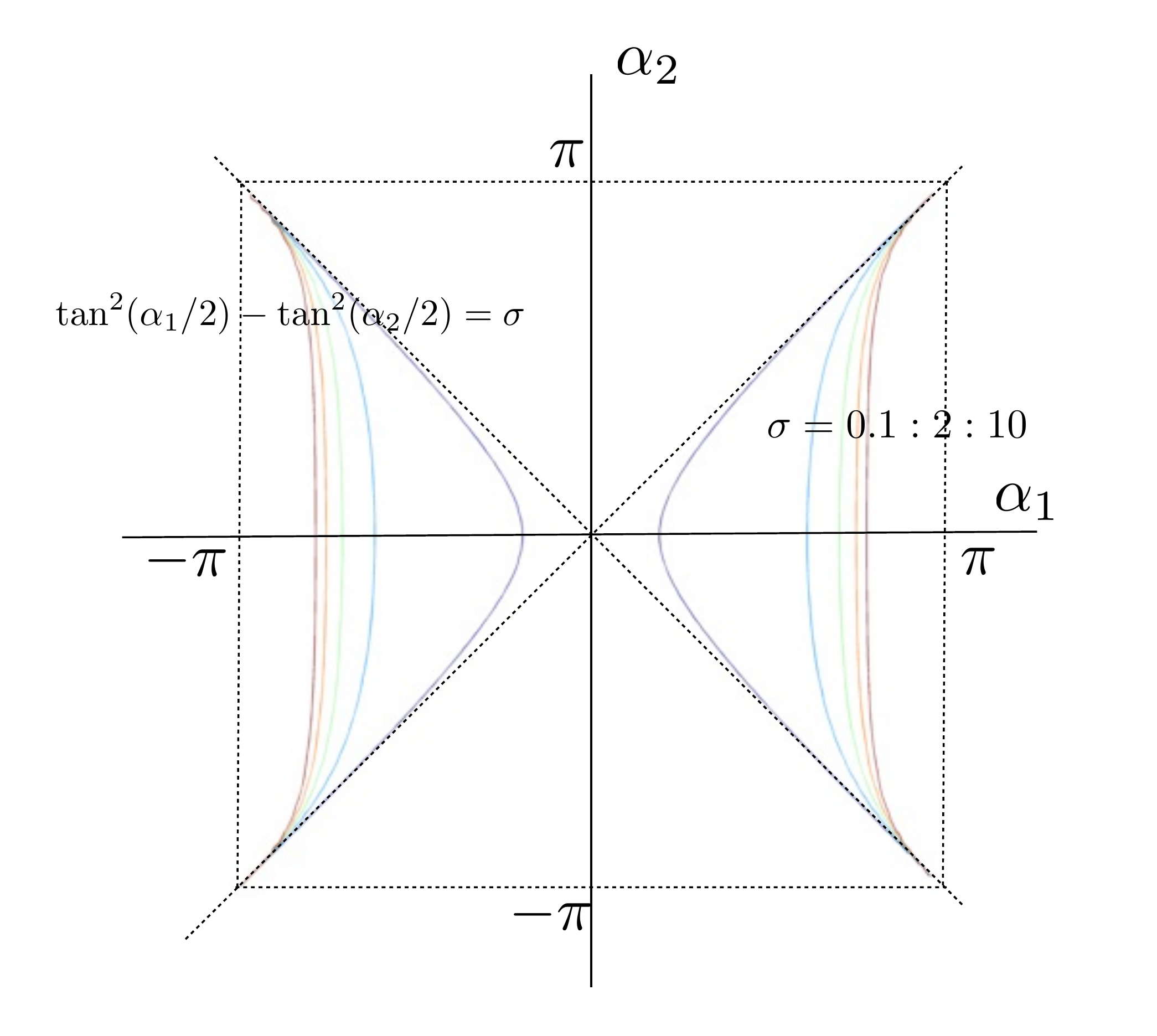}
	\caption{Phase space of self-adjoint extensions of the half-line/half-spin system as function of the spectral gap $\sigma$ possessing a unique ground state.}\label{phasespace}
\end{figure}

%%%%%%%%%%%%%%%%%%%%%%%%%%%%%%%%%%%%%%%%%%%%%%%%%%%%%%%%%%%%%%%%%%%%

The curves defined by eq. \eqref{compatibility} determined by the values of $\sigma$ provide families of non-separable, self-adjoint extensions of $H$ compatible with the structure of $H_B$.   Suppose now that we select as initial state the eigenstate corresponding to the extension defined by $\alpha_1=\arctan\sqrt{\sigma},\;\alpha_2=0$, this is $\Phi_0=e^{-\sqrt{\sigma}/2x}\o\rho_1$. 

Consider now the (time-dependent) Hamiltonian $H$ for the bipartite system given by eq. \eqref{H_half_line} and domain defined by the one-parameter family of self-adjoint extensions defined by the unitary matrices:
	\begin{equation}
		U_{s(t)} =
		\begin{bmatrix}
			e^{2is(t)} & 0 \\ 0 & e^{2is'(t)}
		\end{bmatrix}
	\end{equation}
with $s(t),s'(t)$ such that $\tan^2s-\tan^2s'=\sigma$, this is $s'=\arctan\sqrt{\tan^2s-\sigma}$.    That is, the time-dependence of the evolution of the system is not in the form of the infinitesimal generator but on its domain, which changes with time according with Eq. \eqref{asorey} because the unitary operator $U_{s(t)}$ that defines the domain depends on $t$.

Suppose that we proceed to modify the self-adjoint extension adiabatically.  For that we may choose the parametrization $s = s(t)$, with $t$ the physical time, in such a way that $0< ds/dt << 1$.  Then, in the adiabatic approximation, the eigenstate $\Phi_0$ will change with $t$ but it will remain close to the  ground state of the self-adjoint extension $H_{U_s}$, its (unique) eigenstate, and it will be given by:
	\begin{equation}\label{stationary state}
		\Phi_s=C_1e^{-(\tan s)\; x}\o\rho_1+C_2e^{-(\tan s')\; x}\o\rho_2\;,\quad 0<s<\pi/2\;.
	\end{equation}
Such state $\Phi_s$ is generically an entangled state in $\H_A\cotimes\H_B$.
Notice that the phase diagram of the self-adjoint extensions constructed in this way is periodic, see Figure \ref{periodic}, where black dots correspond to separable states of the form either $e^{-\sqrt{\xi}x}\o\rho_1$ or $e^{-\sqrt{\xi}x}\o\rho_2$.

%%%%%%%%%%%%%%%%%%%%%%%%%%%%%%%%%%%%%%%%%%%%%%%%%%%%%%%%%%%%%%%%%%%%

\begin{figure}[ht]
	\centering
	\includegraphics[height=6cm]{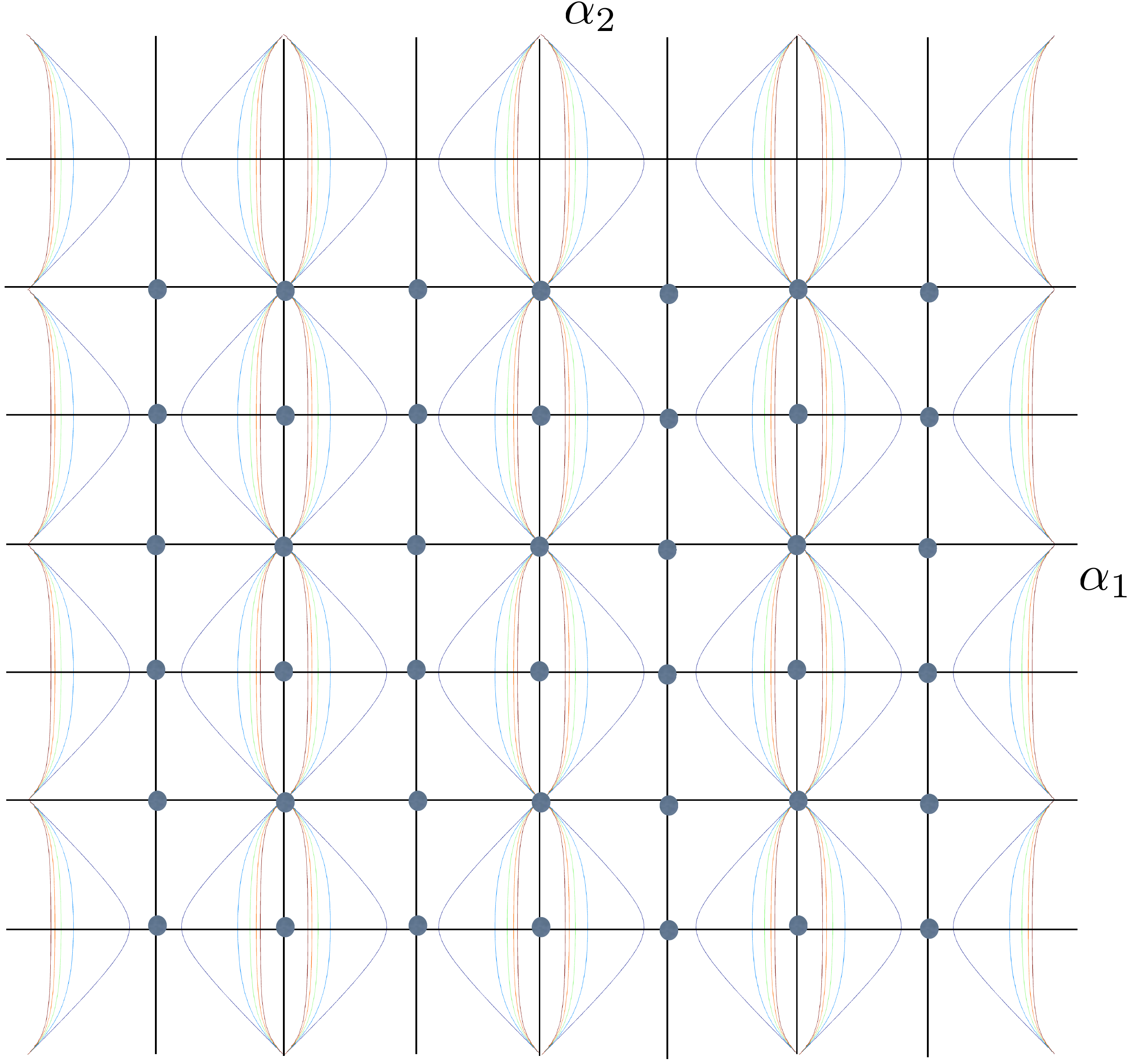}
	\caption{Curves of self-adjoint extensions in the Abelian torus $\mathbb{T}^2 \subset U(2)$ with a single point spectrum.}\label{periodic}
\end{figure}

%%%%%%%%%%%%%%%%%%%%%%%%%%%%%%%%%%%%%%%%%%%%%%%%%%%%%%%%%%%%%%%%%%%%

\subsection*{The half-line/multipartite spin 1/2 system}

We can elaborate the previous example again by considering a system $B$ that is already a composite system, i.e. $\H_B=\H_{B_1}\cotimes\H_{B_2}$ with $\operatorname{dim}\H_{B_\alpha}=n_\alpha$, $\alpha= 1,2$. The self-adjoint operators $H_{B_\alpha}$, $\alpha=1,2$ have eigenvalues $\lambda^{(\alpha)}_{k_\alpha}$, $k_\alpha=1,\dots,n_\alpha$ and a basis of eigenvectors of the operator $H_{B_1}\o\mathbb{I}+\mathbb{I}\o H_{B_2}$ given by
	\begin{equation}
		\rho_{k_1,k_2}=\rho^{(1)}_{k_1}\o\rho^{(2)}_{k_2}\;,
	\end{equation}
where $\rho^{(\alpha)}_{k_\alpha}$ are eigenvectors with eigenvalues $\lambda^{(\alpha)}_{k_\alpha}$. The eigenvalue corresponding to the eigenvector $\rho_{k_1,k_2}$ is just $\lambda^{(1)}_{k_1}+\lambda^{(2)}_{k_2}$. Now we compute the system $A\times B$ to get 
	\begin{equation}
		\H=\H_{A}\hat{\o}\H_{B}\simeq\L^2(\mathbb{R}^+;\H_{B_1}\cotimes\H_{B_2})
	\end{equation}
and we expand $\Phi\in\H$ as
	\begin{equation}
		\Phi=\sum_{1\leq k_{\alpha}\leq n_\alpha}\Phi_{k_1,k_2}\o\rho_{k_1,k_2}\;.
	\end{equation}
In the same way $\varphi=\sum_{1\leq k_{\alpha}\leq n_\alpha}\varphi_{k_1,k_2}\o\rho_{k_1,k_2}$ and $\dot{\varphi}=\sum_{1\leq k_{\alpha}\leq n_\alpha}\dot{\varphi}_{k_1,k_2}\o\rho_{k_1,k_2}$ with $\varphi_{k_1,k_2}=\Phi(0)_{k_1,k_2}$ and $\dot{\varphi}_{k_1,k_2}=-\frac{\d\Phi_{k_1,k_2}}{\d x}\bigr|_{x=0}$.
Finally, we notice that the space of self-adjoint extensions of the composite symmetric operator $H$ is given by $\mathcal{U}(\L^2(\partial\mathbb{R}^+\o\H_{B_1}\o\H_{B_2}))$, i.e.
	\begin{equation}
		\mathcal{M}_{AB} = U (n_1\cdot n_2) = U (N)\;,\quad N=n_1\cdot n_2,\; n_\alpha=\operatorname{dim}\H_{B_\alpha}\;, \alpha=1,2\;.
	\end{equation}
Notice that the assumptions of Theorem \ref{separab} do not hold in this case. However, the intuition provided by Theorem \ref{separab} makes us expect that separable dynamics will correspond to $U=U_A\times\mathbb{I}$, $U_A=e^{i\delta}$. Hence, let us choose boundary conditions leading to non-separable dynamics in the composite system $A\times B$ and in $B$ itself.

We consider for instance $U=\mathbb{I}\times V$ with $V\in \mathcal{U}(\H_{B_1}\o\H_{B_2})$. Again we choose a simplifying hypothesis and assume that the spectrum is non-degenerate. Consider the ordered spectrum of the Hamiltonian $H_B$, i.e. $\Lambda_1=\max\{\lambda^{(1)}_{k_1}+\lambda^{(2)}_{k_2}\}=\lambda^{(1)}_{s_1}+\lambda^{(2)}_{s_2}\geq\Lambda_2=\lambda^{(1)}_{r_1}+\lambda^{(2)}_{r_2}\geq\cdots\geq\Lambda_{N}=\min\{\lambda^{(1)}_{k_1}+\lambda^{(2)}_{k_2}\}$ and let $\Pi_1,\dots,\Pi_{N}$ be the corresponding eigenvectors. Then $H_B\Pi_l=\Lambda_l\Pi_l$.  We choose now the matrix $V$ to be diagonal in this basis, 
$$V=\operatorname{diag}(e^{i\alpha_1},e^{i\alpha_2},\dots,e^{i\alpha_N})$$ 
and repeating the computations performed in the previous example we will get that the point spectrum of the operator $H$ is given by $E=\Lambda_l-\tan^2\alpha_l,\quad l=1,\dots,N$ which imposes $N-1$ conditions on the parameters $\alpha_l$ of the form
	\begin{equation}\label{multi_conditions}
		\Lambda_l-\tan^2\alpha_l=\Lambda_{l+1}-\tan^2\alpha_{l+1}\;,\quad l=1,\dots,N-1\;.
	\end{equation}
The previous equations \eqref{multi_conditions} define a curve in the $N$-dimensional maximal compact abelian subgroup of $U (N)$ similar to those exhibited in Figure \ref{periodic}. Again, a similar analysis as in the example of a single spin 1/2 system allows to conclude that an adiabatic deformation of the system along this curve will take a separable state, for instance $\Phi_{11}\o\rho_{11}=\Phi_{11}\o\rho^{(1)}_1\o\rho^{(2)}_1$ into (maximally) non-separable states.

%%%%%%%%%%%%%%%%%%%%%%%%%%%%%%%%%%%%%%%%%%%%%%%%%%%%%%%%%%%%%%%%%%%%%%%%%%%%%%%%%%%%%%%%%%%%%%%%%%%%%%%%%%%%%%%%%%%%%%%%%%%%%

\section{The quantum planar rotor-spin system}\label{quantum_rotor}

We consider as final example the interesting case of an hybrid system that captures some properties of electron--nucleus systems described recently (see \cite{Sa12}). System $A$ will be now a particle moving in the interval $\Omega_A=[0,1]$ with measure $\d x$, i.e. $\H_A=\L^2([0,1],\d x).$ Unlike in the previous case, now the boundary of system $A$ has two points and therefore the self-adjoint extensions of system $A$ alone are going to be parametrized by matrices in $U(2)$.  All of them have a discrete spectrum (\cite{We80}), so that now we are going to be under the conditions of Theorem \ref{separab}.  Actually, we are going to consider a planar rotor with quasi-periodic boundary conditions \cite{As83}, i.e.,  the previous system with self-adjoint extensions determined by the unitary matrix
	\begin{equation}\label{quasiperiodic}
		U_{A,\delta}=
			\begin{bmatrix}
				0 & e^{i\delta} \\ e^{-i\delta} & 0
			\end{bmatrix}\in \mathcal{U}(\L^2(\partial[0,1]))\;,
	\end{equation}
that correspond to boundary conditions $\Phi(0)=e^{i\delta}\Phi(1)$ and $\Phi'(0)=e^{i\delta}\Phi'(1)$.
Now we will consider as bulk system $B$ a two-level system, for instance a spin $1/2$ system, with dynamics given by $H_B=\mu \sigma_z$, where $\sigma_z$ is the diagonal Pauli matrix 
$$\sigma_z=\begin{bmatrix} 1& 0\\ 0 &-1 \end{bmatrix}$$ 
and $\mu$ is a constant that accounts for both, the coupling constant of the magnetic field with the spin $1/2$ system and the strength of the magnetic field. Then, $\H=\L^2(S^1,\frac{\d x}{2\pi})\hat{\o}\mathbb{C}^2=\L^2(S^1;\mathbb{C}^2)$ is the state space of the total system and we consider
	\begin{equation}\label{planar rotor spin}
		H=-\frac{\d^2}{\d x^2}\o\mathbb{I}+\mathbb{I}\o H_{B}
	\end{equation}
as the total Hamiltonian. For this particular example we turn to the standard notation for spin systems and write the eigenstates corresponding to $H_B$ as  $|\negthickspace\uparrow\downarrow\rangle$. Therefore $H_B|\negthickspace\uparrow\downarrow\rangle=\pm\mu|\negthickspace\uparrow\downarrow\rangle$ and a particular element $\Phi$ of the composite system $\H=\H_A\cotimes\H_B$ will admit the decomposition $\Phi=\Phi^{\uparrow}\o|\negthickspace\uparrow\rangle+\Phi^{\downarrow}\o|\negthickspace\downarrow\rangle$.  
 As boundary conditions we choose $U\in\mathcal{U}(\mathbb{C}_A^2)\o\mathcal{U}(\mathbb{C}_B^2)$ of the form $U=U_{A,\delta}\o U_B$, with $U_{A,\delta}$ as in eq. \eqref{quasiperiodic}.
 
Physically, this system can be interpreted as follows (see Fig. \ref{Quantum Compass}).   There is a charged particle moving along a circle \cite{As83}. In the center of this orbit, there is a fixed spin that interacts with a magnetic field of strength $\mu$ perpendicular to the plane of the orbit. The component $U_B$ of the boundary condition shall be interpreted as a macroscopic interaction \emph{triggered} when the orbiting charged particle traverses an ideally infinitesimal region of the orbit.
\begin{figure}[ht]
	\centering
	\includegraphics[height=4cm]{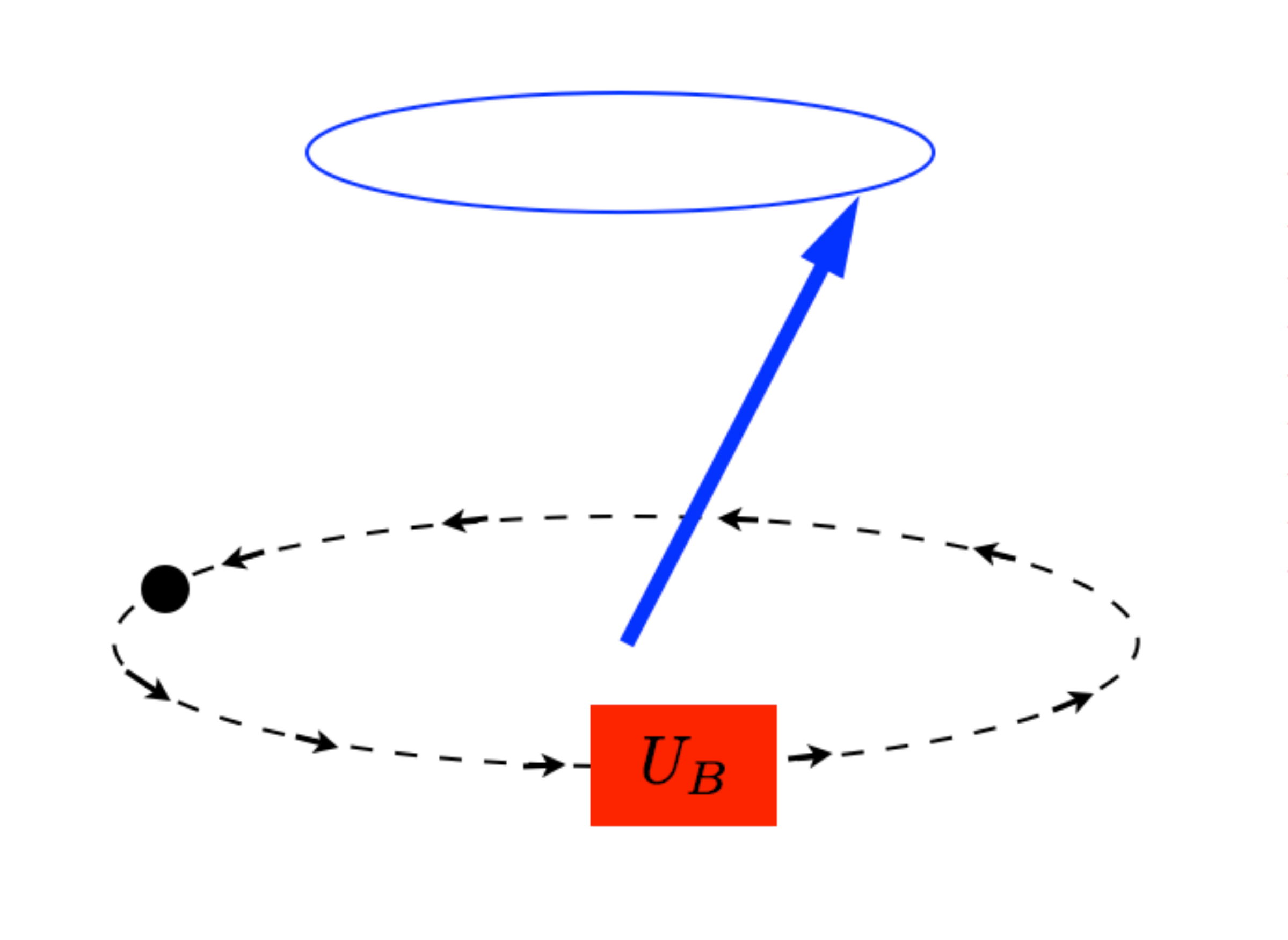}
	\caption{Quantum Compass}\label{Quantum Compass}
\end{figure}

We are going now to consider two different meaningful situations (compare with eq. \eqref{non-separable condition}) for the boundary conditions   corresponding to subsystem $B$. The first situation will correspond to select the unitary matrix $U_B$ diagonal in the basis of $H_B$, namely:  
	\begin{equation}
		U_B=
		\begin{bmatrix}
			e^{i\alpha} & 0 \\ 0 & e^{-i \alpha}
		\end{bmatrix}\;.
	\end{equation} 
 The boundary conditions defined by these unitary matrices take the explicit form
 	$$\begin{array}{c l}
		\Phi^{\up}(0)+i\Phi^{\up\prime}(0)=e^{i(\alpha+\delta)}\bigl(\Phi^{\up}(1)+i\Phi^{\up\prime}(1)\bigr)& \\
		\Phi^{\up}(1)-i\Phi^{\up\prime}(1)=e^{i(\alpha-\delta)}\bigl(\Phi^{\up}(0)-i\Phi^{\up\prime}(0)\bigr) & \\
		\Phi^{\down}(0)+i\Phi^{\down\prime}(0)=e^{i(-\alpha+\delta)}\bigl(\Phi^{\down}(1)+i\Phi^{\down\prime}(1)\bigr) & \\
		\Phi^{\down}(1)-i\Phi^{\down\prime}(1)=e^{i(-\alpha-\delta)}\bigl(\Phi^{\down}(0)-i\Phi^{\down\prime}(0)\bigr) & .
	\end{array}$$
One can proceed like in the previous examples and impose the above boundary conditions to the general solution of the spectral problem, eq. \eqref{planar rotor spin}, given by 
	\begin{align}
		\Phi^{\up}(x)&=Ae^{i\sqrt{E-\mu}x}+Be^{-i\sqrt{E-\mu}x}\notag\\
		\Phi^{\down}(x)&=Ce^{i\sqrt{E+\mu}x}+De^{-i\sqrt{E+\mu}x}\;,\notag	
	\end{align}
to find the corresponding spectral function associated to the problem. In this case one obtains the following spectral function
	\begin{align}
		\sigma_\alpha(E)=\Bigl[2i&\sin(\sqrt{E-\mu})+2iE\sin(\sqrt{E-\mu})  2i\mu\sin(\sqrt{E-\mu})-8\sqrt{E-\mu}\cos(\delta)e^{i\alpha} + \notag\\
		 & + 8\sqrt{E-\mu}\cos(\sqrt{E-\mu})\cos(\alpha)e^{i\alpha}-2i(E-\mu+1)\sin(E-\mu)e^{i2\alpha}\Bigr] \times   \notag\\
                    & \times   \Bigl[2i\sin(\sqrt{E+\mu}) +2iE\sin(\sqrt{E+\mu})+2i\mu\sin(\sqrt{E+\mu})- 8\sqrt{E+\mu}\cos(\delta)e^{-i\alpha} + \notag \\ 
                    & +8\sqrt{E-\mu}\cos(\sqrt{E-\mu})\cos(\alpha)e^{-i\alpha}-2i(E+\mu+1)\sin(E+\mu)e^{-i2\alpha}\Bigr]\;,\notag
	\end{align}
whose zeros are the corresponding eigenvalues.

Finding the zeros of this transcendental function has to be done numerically. However, this task can be challenging, especially because $\sigma_\alpha(E)$ is very close to vanish in some regions. Moreover, the information about the separability of the dynamical evolution depends on the eigenfunctions of the problem as shown in the previous sections. For all these reasons, in order to check that the above problem is not leading to separable dynamics, we will take the approach introduced in \cite{Ib11}.   There, an algorithm based on the Finite Element Method is introduced that is able to solve the spectral problem for any self-adjoint extension of a 1D Schr\"odinger problem. Then it is enough to use the isomorphism $\L^2([0,1])\otimes\mathbb{C}^2\simeq \L^2([0,1])\oplus\L^2([0,1])$ to rewrite the problem given by eq. \eqref{planar rotor spin} into a form that can be handled by this numerical procedure.  Figure \ref{Perdiag eigenfunctions} shows the eigenfunctions corresponding to the 6 smallest energies returned by the algorithm for $\mu=10$, $\delta=\pi/2$, $\alpha=\pi/2$. In each graph are represented simultaneously the two components, $\Phi^{\up}(x)$ and $\Phi^{\down}(x)$, of the eigenfunction $\Phi=\Phi^{\up}|\up\rangle + \Phi^{\down}|\down\rangle$\,. The particular values of the energies are not shown because they are not relevant for the discussion.
\begin{figure}[ht]
	\centering
	\includegraphics[height=6.5cm]{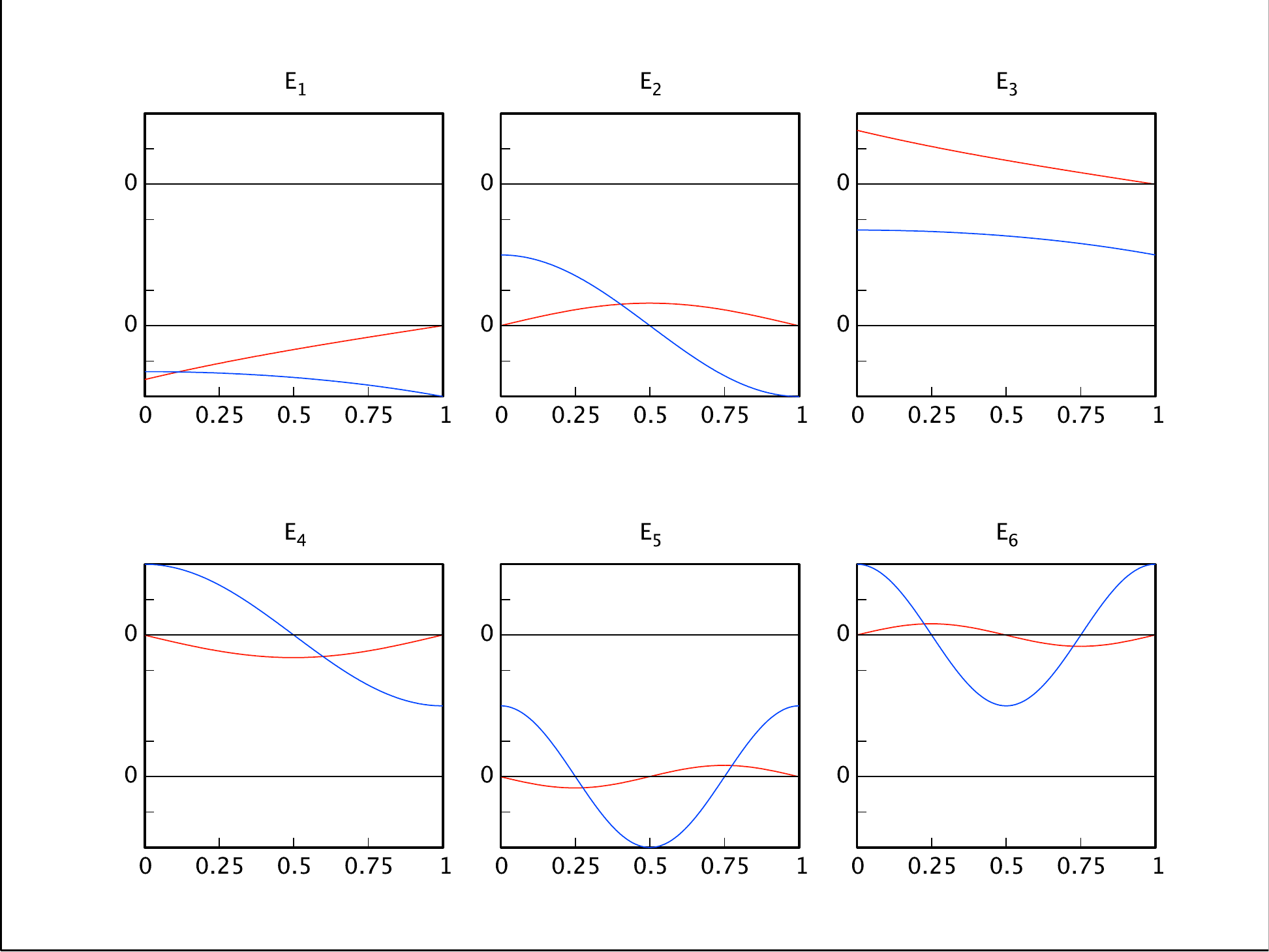}
	\caption{Eigenfunctions of the 6 lowest energy levels for the case $\mu=10$, $\delta=\pi/2$, $\alpha=\pi/2$. On each graph $\Phi^\up(x)$ and $\Phi^\down(x)$ are plotted simultaneously. Real parts are plotted in blue, imaginary parts in red.}\label{Perdiag eigenfunctions}
\end{figure}
As it can be appreciated, the eigenfunctions are separable states in this case. However, separability of the eigenfunctions is not enough to guarantee the separability of the dynamics. According to Section \ref{separable dynamics}, the eigenfunctions of the total Hamiltonian need to admit a factorization $\psi_l\otimes\rho_b$ in terms of the eigenfunctions $\{\psi_l\}$ and $\{\rho_b\}$ of the Hamiltonians of the parties $H_A$ and $H_B$ respectively. In other words, the indices $l$ and $b$ must be independent. As it can be appreciated comparing the eigenfunctions $\Phi_1$ and $\Phi_3$ corresponding to the eigenvalues $E_1$ and $E_3$ respectively, they are not of the form $\Phi_1=\psi(x)\otimes\rho_1$ and $\psi(x)\otimes\rho_2$ for some function $\psi\in\mathcal{L}^2([0,1])$ showing that the set $\{\psi_l\}$ is not independent of $\{\rho_b\}$\,. The same argument holds for the pairs $E_2$, $E_4$ and $E_5$, $E_6$. HEnce we conclude that we have non-separable dynamics for this particular choice of the boundary conditions.\\

Now we consider a different situation where the unitary matrix $U_B$ is taken anti-diagonal with respect to the given basis of $H_B$ and given by
	\begin{equation}
		U_B=
		\begin{bmatrix}
			0 & e^{i\beta}  \\  e^{-i \beta} & 0
		\end{bmatrix}\;.
	\end{equation} 
In this case the boundary conditions defining the system take the form:
 	$$\begin{array}{c l}
		\Phi^{\up}(0)+i\Phi^{\up\prime}(0)=e^{i(\beta+\delta)}\bigl(\Phi^{\down}(1)+i\Phi^{\down\prime}(1)\bigr)& \\
		\Phi^{\up}(1)-i\Phi^{\up\prime}(1)=e^{i(\beta-\delta)}\bigl(\Phi^{\down}(0)-i\Phi^{\down\prime}(0)\bigr) & \\
		\Phi^{\down}(0)+i\Phi^{\down\prime}(0)=e^{i(-\beta+\delta)}\bigl(\Phi^{\up}(1)+i\Phi^{\up\prime}(1)\bigr) & \\
		\Phi^{\down}(1)-i\Phi^{\down\prime}(1)=e^{i(-\beta-\delta)}\bigl(\Phi^{\up}(0)-i\Phi^{\up\prime}(0)\bigr) & .
	\end{array}$$
Again, one can compute the spectral function associated to this problem and we get:
	\begin{align}
		\sigma_\beta(E)&\propto\sqrt{E^2-\mu^2}\cos(\sqrt{E-\mu})\cos(\sqrt{E+\mu})- \notag\\ &-E\sin(\sqrt{E-\mu})\sin(\sqrt{E+\mu})-\sqrt{E-\mu}\sqrt{E+\mu}\cos(2\delta)\;.\notag
	\end{align}
Surprisingly, the spectral function does not depend on the parameter $\beta$ in this case, but the eigenfunctions do.  In Fig. \ref{PerPer eigenfunctions} are plotted the eigenfunctions corresponding to the case $\mu=10$, $\delta=\pi/2$, $\beta=\pi/2$.  One can appreciate that they are non-separable and therefore the dynamics characterized by this last set of boundary conditions is not separable.
\begin{figure}[h]
	\centering
	\includegraphics[height=6.5cm]{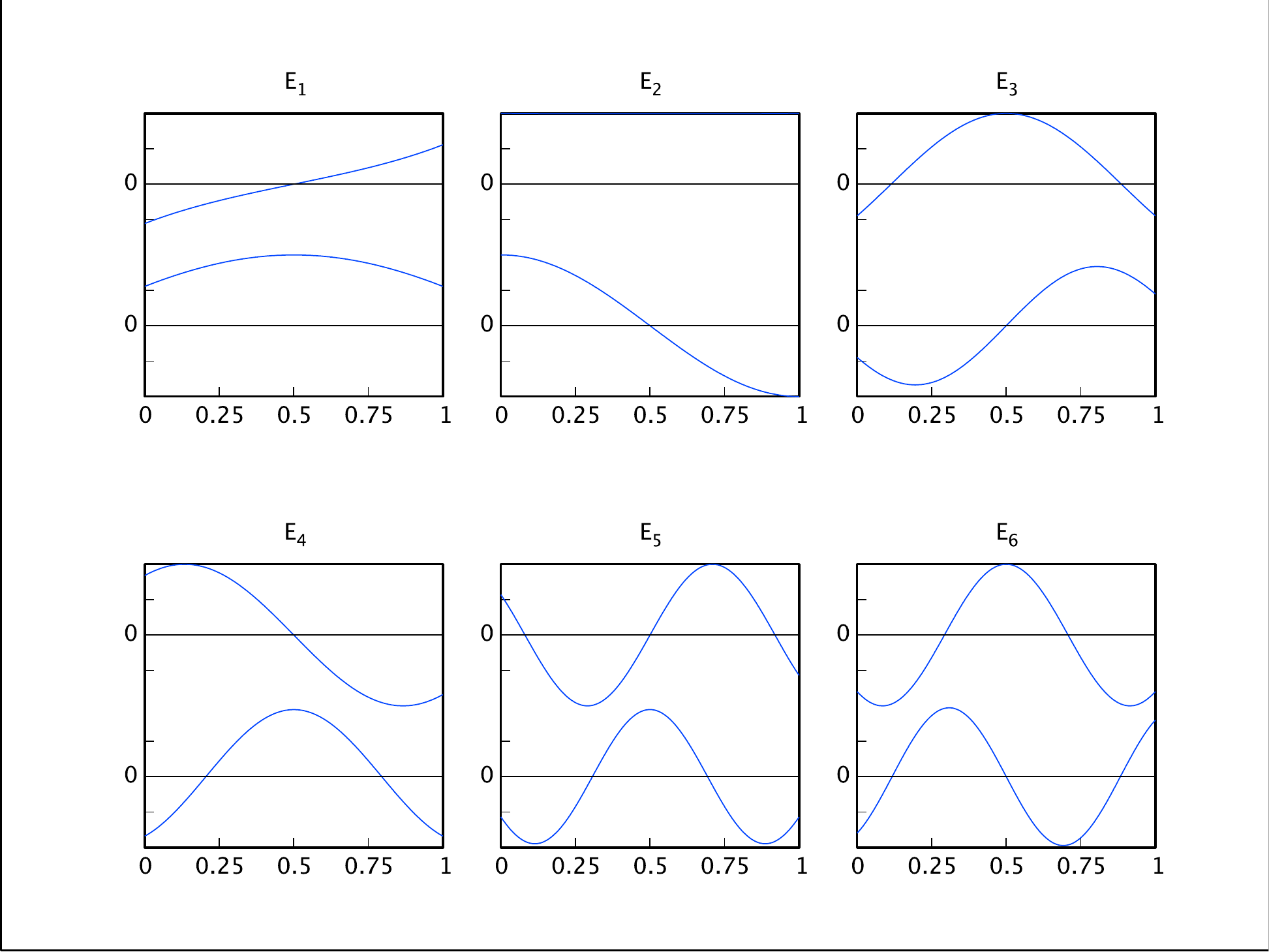}
	\caption{Eigenfunctions of the 6 lowest energy levels for the case $\mu=10$, $\delta=\pi/2$, $\beta=\pi/2$.  On each graph $\Phi^\up(x)$ and $\Phi^\down(x)$ are plotted simultaneously. The imaginary parts vanish identically in this case.}\label{PerPer eigenfunctions}
\end{figure}
%

%%%%%%%%%%%%%%%%%%%%%%%%%%%%%%%%%%%%%%%%%%%%%%%%%%%%%%%%%%%%%%%%%%%%%%%%%%%%%%%%%%%%%%%%%%%%%%%%%%%%%%%%%%%%%%%%%%%%%%%%%%%%

\section{Conclusions and discussion}

Along the article we have shown that manipulating boundary conditions for a class of bipartite systems it is possible to evolve a separable state into an entangled one. We have shown that we can achieve this dynamically, by changing the boundary conditions in a time dependent way, see Section \ref{simple}. This phenomenon also arises for fixed boundary conditions as the examples in Section \ref{quantum_rotor} show. The reason for this phenomenon lies in the existence of many self-adjoint extensions of a bipartite, symmetric system that lead to non-separable dynamics.    We have been able to characterize all boundary conditions leading to separable dynamics in a class of symmetric bipartite systems.

The systems exhibited are hybrid systems and one of the parties, the control or auxiliary system, is symmetric but not self-adjoint. The most remarkable fact about this class of systems is that the space of self-adjoint extensions is much larger than the space of extensions of the standalone control system and it incorporates boundary data that affect simultaneously the control and the controlled or bulk system.  The controlled system has unitary dynamics, but together with the control system it becomes non-separable. Hence, taking the partial trace with respect to system $A$ will not give us back the original dynamics $U^B_t$.
These ideas can be used to generate entangled states in a precise way, or to help to preserve entanglement without actually interacting with the ``bulk'' of the controlled system.
The relation of these ideas with recent work on adiabatic computation and robust entanglement in hybrid systems will be pursued in the future.

\end{document}